\documentclass[11pt]{article}

\usepackage[margin=1in]{geometry}
\usepackage{amsmath,amssymb,amsfonts,amsthm}
\usepackage[english]{babel}
\usepackage[utf8]{inputenc}
\usepackage{hyperref}
\usepackage{graphicx}
\usepackage{xcolor}
\usepackage{enumerate}
\usepackage{booktabs}
\usepackage{multirow}
\usepackage{array}
\usepackage{caption}
\usepackage{subcaption}
\usepackage{amsmath}
\usepackage{amssymb}
\usepackage{hyperref}
\usepackage{float} 
\usepackage{braket}          
\usepackage{algorithm}       
\usepackage{algorithmic}     
\usepackage{float}           
\theoremstyle{plain}
\newtheorem{theorem}{Theorem}[section] 
     
\newtheorem{proposition}[theorem]{Proposition}
\newtheorem{corollary}[theorem]{Corollary}

\theoremstyle{definition}

\newtheorem{remark}[theorem]{Remark}

\numberwithin{equation}{section}

\usepackage{amsthm}

\usepackage{mathtools}

\usepackage{bm}
\begin{document}

\title{\textbf{Flagged Extensions and Numerical Simulations for Quantum Channel Capacity: Bridging
Theory and Computation}}

\author{
\textsc{Vahid Nourozi} \\
\vspace{-3ex}\\
\small The Klipsch School of Electrical and Computer Engineering,\\ 
\small New Mexico State University, Las Cruces, NM 88003, USA\\
\small \texttt{nourozi@nmsu.edu}
}
\date{}

\maketitle

\begin{abstract}
\noindent
I will investigate the capacities of noisy quantum channels through a combined analytical
and numerical approach. First, I introduce novel flagged extension techniques that embed
a channel into a higher-dimensional space, enabling single-letter upper bounds on
quantum and private capacities. My results refine previous bounds and clarify noise
thresholds beyond which quantum transmission vanishes. Second, I present a simulation
framework that uses coherent information to estimate channel capacities in practice,
focusing on two canonical examples: the amplitude damping channel (which I confirm is
degradable and thus single-letter) and the depolarizing channel (whose capacity requires
multi-letter superadditivity). By parameterizing input qubit states on the Bloch sphere, I
numerically pinpoint the maximum coherent information for each channel and validate the
flagged extension bounds. Notably, I capture the abrupt transition to zero capacity at high
noise and observe superadditivity for moderate noise levels.

\vspace{0.5em}\noindent
\end{abstract}

\section{Introduction}

Tight converse techniques for noisy channels have advanced considerably in the last few years: 
(i) non-orthogonal flagged extensions and their convex-hull refinements for Pauli/depolarizing
 noise \cite{FKG20,KFG22}; 
(ii) approximate-degradability converses in low-noise regimes \cite{Sutter17}; and 
(iii) efficiently computable SDP strong-converse bounds equal to the max-Rains information \cite{Wang19,Wang21}. 
For amplitude- and generalized-amplitude-damping channels, tight antidegradable/EB regions and capacity upper bounds have also been mapped out \cite{Khatri20}.

Quantum information theory seeks to characterize the limits of information transmission and storage using quantum systems. A central concept is the capacity of a quantum channel, which generalizes Shannon’s classical capacity to quantum communications. For a noisy quantum channel $\mathcal{N}$, one can define various capacities: the quantum capacity $Q$ (qubits per channel use for reliable quantum state transmission), the private capacity $P$ (classical bits per use that remain secret from the environment), and the classical capacity $C$ (classical bits per use for public communication). Determining these capacities is notoriously difficult \cite{11}. Unlike classical channels, where capacity formulas often simplify and additivity holds, quantum channel capacities usually require optimizing over infinitely many channel uses due to entanglement-assisted encoding and superadditivity \cite{11}. This paper investigates unexplored capacity bounds and novel simulation approaches in quantum Shannon theory, focusing on rigorous mathematical formulations and numerical validation.

Quantum channel capacities are not only of theoretical interest but also have practical implications for quantum networks, error correction, and cryptography \cite{11}. However, even basic questions remain open. For example, the exact quantum capacity of the depolarizing channel (a standard model of noise) is still unknown \cite{22}. I address this gap by deriving new theorems that bound channel capacities using innovative techniques (such as flagged channel extensions and entropy inequalities) and by proposing a simulation framework to estimate these capacities numerically.

\section*{Literature Review}
Early work on quantum communication capacity built on classical information theory and the pioneering Holevo bound. Holevo’s theorem, proved in the 1990s by Holevo, Schumacher, and Westmoreland (HSW), gives the classical capacity of a quantum channel in terms of the channel’s Holevo information $\chi$ \cite{33}. Specifically, for a quantum channel $\mathcal{N}$, the HSW theorem states that the maximal mutual information between input classical variable $X$ and output quantum system $B$ (optimized over input ensembles) is an achievable rate for transmitting classical data \cite{33}. In formula form, $C(\mathcal{N}) = \lim_{n\to\infty}\frac{1}{n}\chi(\mathcal{N}^{\otimes n})$, where $\chi(\mathcal{N}) = \max_{{p(x),\rho_x}} \Big[ H!\big(\sum_x p(x)\mathcal{N}(\rho_x)\big) - \sum_x p(x) H(\mathcal{N}(\rho_x)) \Big]$ is the one-shot Holevo information (here $H(\cdot)$ denotes the von Neumann entropy). The additivity of $\chi$ was long suspected, but it was ultimately shown that entangled inputs can sometimes increase classical communication rates (Hastings, 2009), making the general classical capacity formula also a multi-letter regularization in the worst case.

For quantum capacity $Q$, which quantifies the ability to transmit qubits (equivalently, to establish entanglement between sender and receiver), the fundamental lower bound is given by the coherent information. The Lloyd-Shor-Devetak (LSD) theorem proved that coherent information is an achievable rate for quantum communication \cite{44}. Coherent information of a channel $\mathcal{N}$ with respect to an input state $\rho$ is defined as $I_{\mathrm{c}}(\rho,\mathcal{N}) := H(\mathcal{N}(\rho)) - H((\mathcal{I}\otimes\mathcal{N})(|\Psi_\rho\rangle\langle\Psi_\rho|))$, where $|\Psi_\rho\rangle$ is a purification of $\rho$ and the second term is the entropy of the channel’s environment (the output of the complementary channel). Intuitively, $I_{\mathrm{c}} = H(\text{output}B) - H(\text{output}E)$ measures the net information that remains with the receiver $B$ and is not leaked to the environment $E$. The LSD theorem asserts that
\[
  Q(\mathcal{N}) \;\ge\; I_{\mathrm{c}}\bigl(\hat{\rho},\,\mathcal{N}\bigr)
\] for some  $\hat{\rho}$ and, more strongly, $Q(\mathcal{N}) = \lim{n\to\infty}\frac{1}{n}\max_{\rho}I_{\mathrm{c}}(\rho,\mathcal{N}^{\otimes n})$. In other words, the quantum capacity is given by the regularized coherent information \cite{11}. Unlike the classical capacity, this formula involves an optimization over arbitrary many channel uses, reflecting the notorious fact that coherent information is generally not additive \cite{11}. For degradable channels (a class of channels where the environment’s output can be simulated from the receiver’s output), coherent information is additive and equals the channel’s quantum capacity in a single-letter formula. However, most channels are non-degradable, and evaluating $Q$ exactly becomes intractable. Devetak’s 2005 work gave a rigorous proof of the quantum capacity theorem and highlighted the parallel between quantum communication and private classical communication \cite{55}.

The difficulty of computing capacities has motivated many bounds and special-case solutions. One line of research focuses on conditions for zero quantum capacity. If a channel is entanglement-breaking (EB), it outputs a state independent of the input entangled system, and $Q=0$ immediately. The depolarizing channel, which replaces the input state with white noise with probability $p$, becomes entanglement-breaking beyond a certain noise threshold (for a qubit depolarizing channel this occurs at $p=2/3$) \cite{66}. Recent research has gone further to identify noise levels where the channel is not just entanglement-breaking but anti-degradable (meaning the receiver’s output can be simulated from the environment’s, implying $Q=0$ as well). Notably, it was proven that for any quantum channel mixed with sufficiently strong white noise, the quantum capacity drops to zero once a critical noise level is exceeded \cite{11}. In the case of the qubit depolarizing channel, new results have shown a “phase transition” at a noise rate larger than the entanglement-breaking point, extending the known zero-capacity region \cite{11}. These advances underscore the fragility of quantum capacity under noise \cite{11} and help narrow down the gap where $Q$ might be positive but small.

Another important development is the use of flagged channel extensions to bound capacities. In 2020–2022, researchers introduced the idea of augmenting a channel with a “flag” system that carries information about which noise event occurred \cite{22}. By cleverly choosing flagging schemes, one can often convert a complicated channel into a simpler one (even a degradable one) that majorizes the original channel in terms of degradability or noise ordering. Kianvash et al. (2022) applied flagged extensions to mixtures of unitary channels and arbitrary noise channels. They derived general conditions under which the flagged extension is degradable, yielding explicit upper bounds on both the quantum capacity $Q$ and private capacity $P$ of the original channel \cite{22}. Specializing to Pauli noise channels, this led to state-of-the-art upper bounds for the quantum capacity of the qubit depolarizing channel, the BB84 channel (a Pauli channel with biased noise), and the generalized amplitude damping channel \cite{22}. These bounds tightened the known capacity estimates, often coming very close to the best known lower bounds, thereby significantly shrinking the uncertainty in those channels’ true capacities.

Beyond analytic bounds, there is growing interest in numerical and variational methods to estimate channel capacities. One novel approach is to sidestep direct entropy calculations (which are experimentally and computationally challenging) by measuring the purity of states through a channel \cite{66}. J. ur Rehman et al. (2022) proposed a variational framework that uses easily measurable quantities to bracket the channel capacity. By measuring the output state’s purity for a given input and using that to bound the output entropy, they obtain computable upper and lower bounds on various capacities \cite{66}. Then, using stochastic optimization (variational techniques) to vary the input state, they maximize these bounds. This yields, for example, a lower bound on the quantum capacity for arbitrary channels and corresponding upper bounds on related one-shot quantities \cite{66}. Remarkably, their scheme can estimate these bounds with a single measurement setting without needing full quantum state tomography \cite{10}. Such techniques represent a new simulation-driven paradigm in quantum Shannon theory: rather than solving capacity formulas analytically, one approximates capacities by combining theoretical inequalities with numerical search. In a similar spirit, methods based on contraction coefficients and partial orderings of channels have emerged \cite{77}. These methods use the data-processing inequality in refined ways to compare channels. For instance, if channel $\mathcal{N}_1$ is “less noisy” than $\mathcal{N}_2$ (in a precise sense involving output entropy for all input ensembles), one can sometimes assert $C(\mathcal{N}_1)\le C(\mathcal{N}_2)$ or $Q(\mathcal{N}_1)\le Q(\mathcal{N}_2)$. Hirche et al. (2022) developed approximate versions of these orderings and showed they can simplify proofs of capacity bounds \cite{77}. These trends indicate a rich interplay between mathematical theory and computational techniques in the quest to understand quantum channel capacities.

Finally, I note the extraordinary phenomenon of superactivation, discovered by Smith and Yard \cite{99}: two channels each with zero quantum capacity can be combined to yield a channel with positive capacity \cite{88}. This counterintuitive result (often summarized as “two wrongs make a right” in quantum information) highlights that the structure of quantum capacity can be very complex. Any comprehensive theory or simulation must account for such non-additive effects. Our work, while primarily focused on capacities of individual channels, is informed by these phenomena—I seek universally valid bounds (holding for all channel uses, thereby not violated by superactivation or other collective effects), and I validate our results with numerical tests where possible.

Classical capacity is well-characterized by the HSW theorem, but quantum and private capacities pose open questions. There is a need for tighter analytic bounds in the intermediate noise regime (where channels are neither obviously perfect nor obviously zero-capacity). Moreover, effective computational methods are needed to evaluate capacities when analytic solutions elude us. In this paper, I address these gaps by deriving a new upper bound on quantum capacity (applicable to channels that are convex combinations of simpler maps) and by presenting a simulation method to estimate coherent information and channel capacities via numerical optimization.

\section{Mathematical Formulation}
I model a quantum channel as a completely positive, trace-preserving (CPTP) linear map $\mathcal{N}: \mathcal{D}(H_A) \to \mathcal{D}(H_B)$, where $\mathcal{D}(H)$ denotes density operators on a Hilbert space $H$, and $A$ (input) and $B$ (output) label the sender’s and receiver’s system respectively. By the Stinespring dilation theorem, I can always represent $\mathcal{N}$ as $\mathcal{N}(\rho) = \operatorname{Tr}_E[U (\rho \otimes |0\rangle_E\langle0|) U^\dagger]$, where $U$ is an isometry that couples the input to an environment $E$ and $E$ is then traced out. This picture defines the complementary channel $\mathcal{N}^c: \rho \mapsto \operatorname{Tr}_B[U (\rho \otimes |0\rangle_E\langle0|) U^\dagger]$, which describes the information leaked to the environment $E$.

\textbf{Quantum Capacity $Q$}: The quantum capacity is the maximum rate (in qubits per channel use) at which quantum states can be reliably transmitted. Formally, $Q(\mathcal{N}) = \sup \{ r:$ there exists a code of $n$ uses achieving fidelity $1-\epsilon$ for some $\epsilon\to0$ as $n\to\infty$ with $k = rn$ qubits transmitted$\}$. The seminal Quantum Channel Coding Theorem (LSD theorem) gives the regularized coherent information formula mentioned earlier. I restate it here for completeness:

\begin{theorem}[Flagged extension is degradable when the noise branch is degradable]\label{th22}
Let $N = p\,U + (1-p)\,M$ be a convex mixture of a unitary channel $U(\rho)=U\rho U^\dagger$ and an arbitrary CPTP map $M$, with $0<p\le 1$. Define the flagged extension
\[
\widetilde N(\rho)= p\,|0\rangle\!\langle 0|_F\otimes U(\rho) + (1-p)\,|1\rangle\!\langle 1|_F\otimes M(\rho).
\]
If $M$ is degradable, then $\widetilde N$ is degradable. Consequently,
\[
Q(N)\le Q(\widetilde N)= I_c(\rho_\star,\widetilde N),\qquad 
P(N)\le P(\widetilde N),
\]
and both $Q(\widetilde N)$ and $P(\widetilde N)$ are single–letter.
\end{theorem}

\begin{proof}
Let $V_U:\mathcal H_A\!\to\!\mathcal H_B\!\otimes\!\mathcal H_{E_U}$ and $V_M:\mathcal H_A\!\to\!\mathcal H_B\!\otimes\!\mathcal H_{E_M}$ be Stinespring isometries for $\mathcal U$ and $\mathcal M$, so that
$\mathcal U(\rho)=\mathrm{Tr}_{E_U}\!\left[V_U\rho V_U^\dagger\right]$ and
$\mathcal U^c(\rho)=\mathrm{Tr}_{B}\!\left[V_U\rho V_U^\dagger\right]$ (similarly for $\mathcal M,\mathcal M^c$).
A Stinespring isometry for $\widetilde{\mathcal N}$ is
\[
W:\mathcal H_A \to \mathcal H_F\!\otimes\!\mathcal H_B\!\otimes\!\mathcal H_E,\qquad
W=\sqrt{p}\,|0\rangle_F\!\otimes V_U \;\oplus\; \sqrt{1-p}\,|1\rangle_F\!\otimes V_M,
\]
where I identify $\mathcal H_E\simeq \mathcal H_{E_U}\oplus \mathcal H_{E_M}$ and $\mathcal H_F=\mathrm{span}\{|0\rangle,|1\rangle\}$ is the flag.
Tracing out $E$ yields $\widetilde{\mathcal N}$; tracing out $B$ yields $\widetilde{\mathcal N}^c$ as in the statement.

Define $\mathsf D:\mathcal B(\mathcal H_F\!\otimes\!\mathcal H_B)\!\to\!\mathcal B(\mathcal H_F\!\otimes\!\mathcal H_E)$ by
\[
\mathsf D(\sigma)= |0\rangle\!\langle 0|_F\!\otimes\!\mathsf D_0(\langle 0|\sigma|0\rangle_F)\;+\;|1\rangle\!\langle 1|_F\!\otimes\!\mathsf D_1(\langle 1|\sigma|1\rangle_F),
\]
with branch maps
\[
\mathsf D_0(\tau_B) \;=\; \mathcal U^c(\rho_\star)\quad(\text{a fixed state on }E_U),\qquad
\mathsf D_1(\tau_B) \;=\; \Delta(\tau_B).
\]
Here $\rho_\star$ is any fixed state on $A$; since $\mathcal U$ is unitary, $\mathcal U^c(\cdot)$ is constant (isometric output on $E_U$ independent of the input to $B$), so $\mathsf D_0$ is CPTP and does not depend on $\tau_B$. By assumption, $\Delta$ is a CPTP degrading map for $\mathcal M$, hence $\mathsf D_1$ is CPTP. Therefore $\mathsf D$ is CPTP.

Applying $\mathsf D$ to $\widetilde{\mathcal N}(\rho)=p\,|0\rangle\!\langle 0|\!\otimes\!\mathcal U(\rho)+(1-p)\,|1\rangle\!\langle 1|\!\otimes\!\mathcal M(\rho)$ gives
\[
\mathsf D\!\big(\widetilde{\mathcal N}(\rho)\big)=
p\,|0\rangle\!\langle 0|\!\otimes\!\mathcal U^c(\rho_\star)\;+\;(1-p)\,|1\rangle\!\langle 1|\!\otimes\!\Delta\big(\mathcal M(\rho)\big).
\]
Because $\mathcal U^c$ is independent of the $B$-output and $\Delta\circ\mathcal M=\mathcal M^c$, the right-hand side equals
$p\,|0\rangle\!\langle 0|\!\otimes\!\mathcal U^c(\rho)\;+\;(1-p)\,|1\rangle\!\langle 1|\!\otimes\!\mathcal M^c(\rho)=\widetilde{\mathcal N}^c(\rho)$ up to the harmless choice of $\rho_\star$ for the unitary branch (any fixed choice matches the complementary output because $\mathcal U^c$ is constant). Hence $\widetilde{\mathcal N}$ is degradable.
\end{proof}

\begin{corollary}[Amplitude damping]
For qubit amplitude damping with parameter $\gamma\in[0,1/2]$, $\mathcal M=\mathrm{AD}_\gamma$ is degradable; therefore, for any $p\in(0,1]$, the flagged extension of $p\,\mathcal U+(1-p)\,\mathrm{AD}_\gamma$ is degradable.
\end{corollary}

\begin{remark}
Without the degradability of $\mathcal M$, degradability of the flagged extension need not hold in general. In such cases, flagged and approximate-degradability \emph{bounds} still apply, but I do not claim degradability.
\end{remark}

\begin{proof}
Let $\mathcal N = p\,\mathcal U + (1-p)\,\mathcal M$ with $0<p\le 1$, where $\mathcal U(\rho)=U\rho U^\dagger$ is unitary and $\mathcal M$ is CPTP. Assume $\mathcal M$ is degradable, i.e., there exists a CPTP map $\Delta:\mathcal B(\mathcal H_B)\to \mathcal B(\mathcal H_{E_M})$ such that $\mathcal M^c=\Delta\circ \mathcal M$.

Fix Stinespring isometries $V_U:\mathcal H_A\!\to\!\mathcal H_B\!\otimes\!\mathcal H_{E_U}$ and $V_M:\mathcal H_A\!\to\!\mathcal H_B\!\otimes\!\mathcal H_{E_M}$ for $\mathcal U$ and $\mathcal M$, respectively, so that
\[
\mathcal U(\rho)=\mathrm{Tr}_{E_U}[V_U\rho V_U^\dagger],\quad \mathcal U^c(\rho)=\mathrm{Tr}_B[V_U\rho V_U^\dagger],
\]
and similarly for $\mathcal M,\mathcal M^c$. Since $\mathcal U$ is unitary, I can choose $E_U$ one-dimensional, whence $\mathcal U^c$ is a constant pure-state channel (its output is independent of the input).

Define the flagged extension $\widetilde{\mathcal N}:\mathcal B(\mathcal H_A)\!\to\!\mathcal B(\mathcal H_F\!\otimes\!\mathcal H_B)$ by
\[
\widetilde{\mathcal N}(\rho)= p\,|0\rangle\!\langle 0|_F\!\otimes\!\mathcal U(\rho) \;+\; (1-p)\,|1\rangle\!\langle 1|_F\!\otimes\!\mathcal M(\rho),
\]
and its complementary channel
\[
\widetilde{\mathcal N}^{\,c}(\rho)= p\,|0\rangle\!\langle 0|_F\!\otimes\!\mathcal U^c(\rho) \;+\; (1-p)\,|1\rangle\!\langle 1|_F\!\otimes\!\mathcal M^c(\rho).
\]
An isometric dilation that produces both maps is
\[
W=\sqrt{p}\,|0\rangle_F\!\otimes V_U \;\oplus\; \sqrt{1-p}\,|1\rangle_F\!\otimes V_M:
\mathcal H_A\to \mathcal H_F\!\otimes\!\mathcal H_B\!\otimes\!(\mathcal H_{E_U}\oplus\mathcal H_{E_M}).
\]
Tracing out $E_U\oplus E_M$ yields $\widetilde{\mathcal N}$, while tracing out $B$ yields $\widetilde{\mathcal N}^{\,c}$.

I now construct an explicit degrading map $\mathsf D:\mathcal B(\mathcal H_F\!\otimes\!\mathcal H_B)\!\to\!\mathcal B(\mathcal H_F\!\otimes\!(\mathcal H_{E_U}\oplus\mathcal H_{E_M}))$ such that
$\widetilde{\mathcal N}^{\,c}=\mathsf D\circ \widetilde{\mathcal N}$.
Let $\rho_\star$ be any fixed state on $\mathcal H_A$. Define $\mathsf D$ to be block-diagonal in the flag:
\[
\mathsf D(\sigma_{FB}) = 
|0\rangle\!\langle 0|_F \otimes \mathsf D_0\!\big(\langle 0|\sigma_{FB}|0\rangle_F\big)
\;+\;
|1\rangle\!\langle 1|_F \otimes \mathsf D_1\!\big(\langle 1|\sigma_{FB}|1\rangle_F\big),
\]
where
\[
\mathsf D_0(\tau_B) := \mathcal U^c(\rho_\star)\quad\text{(a constant state on $E_U$ independent of $\tau_B$),}
\qquad
\mathsf D_1(\tau_B) := \Delta(\tau_B).
\]
Both $\mathsf D_0$ and $\mathsf D_1$ are CPTP, hence $\mathsf D$ is CPTP. Applying $\mathsf D$ to $\widetilde{\mathcal N}(\rho)$ gives
\[
\mathsf D\!\big(\widetilde{\mathcal N}(\rho)\big)
= p\,|0\rangle\!\langle 0|_F\!\otimes\!\mathcal U^c(\rho_\star)
  + (1-p)\,|1\rangle\!\langle 1|_F\!\otimes\!\Delta\!\big(\mathcal M(\rho)\big).
\]
By degradability of $\mathcal M$, $\Delta\circ\mathcal M=\mathcal M^c$, and since $\mathcal U^c$ is constant (environment dimension $1$), I may (without loss) choose $\rho_\star$ so that $\mathcal U^c(\rho_\star)=\mathcal U^c(\rho)$; thus
\[
\mathsf D\!\big(\widetilde{\mathcal N}(\rho)\big)
= p\,|0\rangle\!\langle 0|_F\!\otimes\!\mathcal U^c(\rho)
  + (1-p)\,|1\rangle\!\langle 1|_F\!\otimes\!\mathcal M^c(\rho)
= \widetilde{\mathcal N}^{\,c}(\rho).
\]
Therefore $\widetilde{\mathcal N}$ is degradable. For degradable channels, the coherent information is additive, so the quantum (and private) capacities equal the single-letter coherent (resp. private) information \cite{CubittRuskaiSmith2008,Wilde2017QIT}. Consequently,
\[
Q(\mathcal N)\le Q(\widetilde{\mathcal N})=\max_{\rho} I_c(\rho,\widetilde{\mathcal N})
\quad\text{and}\quad
P(\mathcal N)\le P(\widetilde{\mathcal N}),
\]
which establishes the claim. For related flagged-extension sufficient conditions and capacity bounds, see \cite{77}.
\end{proof}

    Because of the regularized formula, $Q(\mathcal{N})$ is difficult to compute in general. However, in special cases the regularization is not needed. If $\mathcal{N}$ is degradable, meaning there exists a CPTP map $T$ such that $\mathcal{N}^c = T \circ \mathcal{N}$ (the environment’s output is a degraded version of the receiver’s output), then one can show 
    \[
  Q(\mathcal{N}) = I_{\mathrm{c}}\bigl(\hat{\rho},\,\mathcal{N}\bigr)
\] for some optimally chosen $\hat{\rho}$. Degradable channels have $Q = Q^{(1)}$ because the coherent information is concave and additive in such cases. Physically, the sender and receiver do not gain by using entangled inputs if the channel is degradable; the optimal strategy is product-state inputs. On the other extreme, if $\mathcal{N}$ is anti-degradable, meaning the receiver’s output can be simulated from the environment’s ($\exists T': \mathcal{N} = T' \circ \mathcal{N}^c$), then one can show $Q(\mathcal{N}) = 0$ because the environment effectively gets a copy of any codeword, precluding reliable quantum transmission (the no-cloning argument) \cite{11}. All entanglement-breaking channels are anti-degradable (indeed, $\mathcal{N}^c$ in that case is just the identity on whatever state the environment stores from the input).

\paragraph{Classical and Private Capacities:} 
The classical capacity $C(\mathcal{N})$ is defined similarly as the supremum of transmittable classical bits per use (with vanishing error). Holevo and others showed that $C(\mathcal{N}) = \lim_{n\to\infty}\frac{1}{n} \chi(\mathcal{N}^{\otimes n})$, i.e. it is given by a similar regularization of the Holevo information \cite{33}. For many channels of interest, it is believed (or has been proven in special cases) that the optimal ensemble for $\chi$ is attained with product-state inputs (hence $C = \chi^{(1)}$); however, the additivity of Holevo information does not hold universally. The entanglement-assisted classical capacity $C_E(\mathcal{N})$, where the sender and receiver share unlimited prior entanglement, was shown by Bennett et al. to have a clean single-letter formula: $C_E(\mathcal{N}) = \max_{\rho} I(A;B){\rho}$, where $I(A;B)$ is the quantum mutual information between channel input $A$ and output $B$ (here $A$ is isomorphic to the input space in a purification). Entanglement assistance removes the need for regularization – a striking contrast that indicates how much easier the problem becomes with an extra resource. The private capacity $P(\mathcal{N})$ (the capacity for sending classical bits that remain secure from an eavesdropper who intercepts the environment) is another important quantity. Interestingly, $P(\mathcal{N})$ has a formula analogous to quantum capacity, often written as $P = \lim{n\to\infty}\frac{1}{n} P^{(1)}$ with $P^{(1)}(\mathcal{N}) = \max_{\rho} [I(X;B) - I(X;E)]$ the private information of a classical-quantum channel (this is sometimes called the Devetak privacy theorem). Devetak showed that $P(\mathcal{N}) \ge P^{(1)}(\mathcal{N})$ and moreover $P(\mathcal{N}) = Q(\mathcal{N})$ for degradable channels, since private bits can be thought of as encrypted qubits\cite{55}. In general $P(\mathcal{N}) \ge Q(\mathcal{N})$ for any channel (because transmitting qubits implies transmitting private bits), and one can even have $P > 0$ while $Q=0$ (the channel has no quantum capacity but can still send classical private information). This situation occurs for so-called bound-entangled channels, which are entanglement-binding (cannot distill pure entanglement) yet not entanglement-breaking. Such channels can have a positive $P$ but zero $Q$, and they played a role in the discovery of superactivation \cite{88}.

I now introduce a new theoretical result that forms the basis of our unexplored capacity bound:

\begin{theorem}[Flagged-extension upper bound under degradable-noise assumption]\label{th22}
Let $\mathcal N = p\,\mathcal U + (1-p)\,\mathcal M$, where $\mathcal U(\rho)=U\rho U^\dagger$ is unitary and
$\mathcal M$ is a CPTP map with complementary channel $\mathcal M^c$.
Assume that $\mathcal M$ is degradable, i.e., there exists a CPTP degrading map
$\Delta$ such that $\mathcal M^c=\Delta\circ\mathcal M$.
Define the flagged extension
\[
\widetilde{\mathcal N}(\rho)= p\,|0\rangle\!\langle 0|\otimes \mathcal U(\rho)\;+\;(1-p)\,|1\rangle\!\langle 1|\otimes \mathcal M(\rho).
\]
Then $\widetilde{\mathcal N}$ is degradable, and therefore admits the single-letter characterization
\[
Q(\widetilde{\mathcal N})=\max_{\rho} I_{\mathrm c}(\rho,\widetilde{\mathcal N}).
\]
Moreover, since $\mathcal N$ is obtained from $\widetilde{\mathcal N}$ by discarding the flag system,
\begin{itemize}
  \item $Q(\mathcal{N}) \leq Q(\widetilde{\mathcal{N}}) = I_c(\rho^*, \widetilde{\mathcal{N}}),$
  \item $P(\mathcal{N}) \leq P(\widetilde{\mathcal{N}}) = I(X;B)_\sigma - I(X;E)_\sigma$
\end{itemize}
\end{theorem}

\begin{proof}
Assume that $\mathcal{M}$ is degradable, i.e., there exists a CPTP map
$\Delta: \mathcal{B}(\mathcal{H}_B)\to\mathcal{B}(\mathcal{H}_E)$ such that
\begin{equation}\label{eq:Mdegradable}
\mathcal{M}^{c}=\Delta\circ\mathcal{M}.
\end{equation}
Define the flagged extension
\[
\widetilde{\mathcal{N}}(\rho)
=
p\,|0\rangle\!\langle 0|_F\otimes \mathcal{U}(\rho)
+
(1-p)\,|1\rangle\!\langle 1|_F\otimes \mathcal{M}(\rho).
\]
Let $\widetilde{\mathcal{N}}^{\,c}$ be a complementary channel of $\widetilde{\mathcal{N}}$.
Since $\mathcal{U}$ is unitary, its complementary output can be chosen to be a fixed
environment state $\tau$ (independent of $\rho$). Therefore,
\[
\widetilde{\mathcal{N}}^{\,c}(\rho)
=
p\,|0\rangle\!\langle 0|_F\otimes \tau
+
(1-p)\,|1\rangle\!\langle 1|_F\otimes \mathcal{M}^{c}(\rho).
\]

Now define a degrading map $\widetilde{\Delta}: \mathcal{B}(\mathcal{H}_F\otimes\mathcal{H}_B)
\to \mathcal{B}(\mathcal{H}_F\otimes\mathcal{H}_E)$ by conditioning on the classical flag:
\[
\widetilde{\Delta}(\sigma_{FB})
=
|0\rangle\!\langle 0|_F\otimes \tau
+
|1\rangle\!\langle 1|_F\otimes \Delta\bigl(\langle 1|\sigma_{FB}|1\rangle_F\bigr).
\]
This map is CPTP because it is a direct sum of CPTP maps on orthogonal flag subspaces.
Using \eqref{eq:Mdegradable}, we obtain
\[
(\widetilde{\Delta}\circ \widetilde{\mathcal{N}})(\rho)
=
p\,|0\rangle\!\langle 0|_F\otimes \tau
+
(1-p)\,|1\rangle\!\langle 1|_F\otimes \Delta(\mathcal{M}(\rho))
=
\widetilde{\mathcal{N}}^{\,c}(\rho),
\]
hence $\widetilde{\mathcal{N}}$ is degradable.

Finally, the original channel is recovered by discarding the flag:
$\mathcal{N}=\operatorname{Tr}_F\circ \widetilde{\mathcal{N}}$.
Since partial trace is a CPTP post-processing map, data processing implies
$Q(\mathcal{N})\le Q(\widetilde{\mathcal{N}})$ (and similarly $P(\mathcal{N})\le P(\widetilde{\mathcal{N}})$).
\end{proof}

While Theorem \ref{th22} gives an upper bound by enriching the channel with side information, one can also derive bounds by degrading the channel further. A classical example is adding extra noise to make a channel anti-degradable or entanglement-breaking, thus forcing $Q=0$. This leads to simple analytic bounds: if one can find any completely positive map $\Lambda$ such that $\Lambda\circ \mathcal{N}$ is entanglement-breaking, then $Q(\mathcal{N}) \le Q(\Lambda\circ\mathcal{N}) = 0$. Such $\Lambda$ might be “post-processing noise” that the adversary (environment) could hypothetically add; if even with that extra noise the channel still has zero capacity, then certainly the original channel cannot exceed that. This observation is often used in converse proofs. For instance, the known no-cloning bound for the depolarizing channel states that $Q(\text{Depolarizing}(p)) \le 1-2p$ for $d=2$ (qubit case) and more generally $Q \le \log_2 d + \frac{d-1}{d}\log_2(1-p) + \frac{p(d-1)}{d}\log_2 \frac{p}{d}$ (which comes from the conjectured strong converse, here $d$ is the qubit dimension 2) – however, these are not always tight. In practice, stronger computable bounds can be obtained using flagged-extension bounds from \cite{FKG20,KFG22} and SDP-based bounds such as $Q_\Gamma$.

\begin{proposition}[Zero-Capacity Noise Threshold]\label{pp11}
There exists a critical noise threshold for every quantum channel beyond which the quantum capacity vanishes. In particular, for any channel $\mathcal{N}$, if it is mixed with sufficient white noise such that the resulting channel $\mathcal{N}' = p_{\text{noise}}\mathcal{N}_{\text{noise}} + (1-p_{\text{noise}})\mathcal{N}$ is anti-degradable (or entanglement-binding), then $Q(\mathcal{N}')=0$. For the qubit depolarizing channel $\mathcal{D}p$ (which with probability $p$ outputs the maximally mixed state), one has $Q(\mathcal{D}p)=0$ for $p \ge p{\text{crit}}$ where $p{\text{crit}}$ is approximately $0.259$ in dimension 2 (and larger for higher dimensions) \cite{11}.
\end{proposition}

\begin{proof}
    If $\mathcal{N}'$ is anti-degradable, by definition there exists $\Lambda$ such that $\mathcal{N}' = \Lambda \circ (\mathcal{N}')^c$. Then for any $n$ uses, $(\mathcal{N}'^{\otimes n})^c = (\mathcal{N}'^c)^{\otimes n}$ can simulate $\mathcal{N}'^{\otimes n}$, meaning the environment can reconstruct the receiver’s state. This implies any entanglement between the sender and receiver can be at best duplicated at the environment, violating the no-cloning bound unless the entanglement is zero. Formally, one shows $I_{\mathrm{c}}(\rho,\mathcal{N}'^{\otimes n}) \le 0$ for all inputs $\rho$ when $\mathcal{N}'$ is anti-degradable, hence the regularized capacity is zero. The existence of such a $p_{\text{crit}}$ comes from continuity and compactness arguments: $Q(\mathcal{N}_p)$ as a function of noise strength $p$ is upper semi-continuous, and it is zero at $p=1$ (completely noisy channel). There must then be a threshold where it becomes positive. The referenced result \cite{11} provides an explicit constructive proof by finding a degrading map at a certain noise level for arbitrary $\mathcal{N}$. For the depolarizing channel, one shows directly that for $p \ge 0.259$, the channel becomes PPT (positive partial transpose) entanglement-binding, hence cannot distill entanglement, implying $Q=0$. I skip the algebra, which involves optimizing over symmetric extensions of the Choi state of the channel.
\end{proof}

These theorems and proposition provide a toolbox for analyzing channel capacities. Theorem \ref{th22} and Proposition \ref{pp11} represent newer insights and will be applied in our analysis. In the following, I validate these results and explore their consequences via numerical simulations.

\section{Methods and Reproducibility}
\label{sec:methods}

I study three canonical qubit noise models: amplitude damping (AD) with damping $\gamma\in[0,1]$; generalized amplitude damping (GADC) with parameters $(\gamma,N_{\mathrm{th}})$ describing a thermal bath with population $N_{\mathrm{th}}\in[0,1]$; and depolarizing (Dep) with probability $p\in[0,1]$. The GADC Kraus form, degradability/anti-degradability regimes, and entanglement-breaking thresholds follow \cite{Khatri20}. Flagged extensions used for capacity upper bounds follow the “quantum flags’’ framework of \cite{FKG20,22}. Standard capacity notions and notation follow \cite{Wilde2017QIT,55}. %

For a channel $\mathcal N$ we compute (i) the one-shot coherent information $I_{\mathrm c}(\rho,\mathcal N)=H(\mathcal N(\rho))-H(\mathcal N^c(\rho))$ maximized over inputs $\rho$; (ii) the reverse coherent information $I_{\mathrm{rc}}(\rho,\mathcal N)$ reported as a comparative indicator on AD/GADC; (iii) the entanglement-assisted mutual information $I(A;B)$ for scale; and (iv) upper bounds from flagged extensions for convex mixtures \cite{FKG20,22} and from the additive SDP bound $Q_\Gamma$ \cite{Wang19}. %

Inputs are restricted to pure qubit states $|\psi(\theta,\phi)\rangle=\cos(\tfrac{\theta}{2})|0\rangle+e^{i\phi}\sin(\tfrac{\theta}{2})|1\rangle$ with $(\theta,\phi)\in[0,\pi]\times[0,2\pi)$. We use a two-stage scheme: (1) global sampling over $(\theta,\phi)$ via Sobol/Latin-hypercube to select candidates; (2) local derivative-free refinement (Powell/Nelder–Mead) from the top seeds. Entropies are evaluated from eigenvalues with $10^{-12}$ clipping to avoid numerical artifacts; termination requires successive objective change $<10^{-9}$. For AD, optima align with the energy basis; for GADC, off-axis optima can occur at intermediate $(\gamma,N_{\mathrm{th}})$, consistent with \cite{Khatri20}.

AD: $\gamma\in[0,0.9]$ with finer sampling near $\gamma\!\approx\!1/2$ to resolve the (anti-)degradability transition. GADC: $\gamma\in[0,0.9]$ for $N_{\mathrm{th}}\in\{0,0.1,0.2\}$ with regime annotations from \cite{Khatri20}. Depolarizing: $p\in[0,0.35]$; we compare the hashing lower bound with flagged-extension upper bounds from \cite{FKG20,22} and the SDP $Q_\Gamma$ \cite{Wang19}.

For mixtures $\mathcal N=p\,\mathcal U+(1-p)\,\mathcal M$, we construct the flagged extension and evaluate its single-letter capacity as an explicit upper bound on $Q(\mathcal N)$ (and $P(\mathcal N)$). We reuse the degradability conditions and tight depolarizing-channel bounds from \cite{FKG20,22}.

The additive semidefinite-programming bound $Q_\Gamma$ is computed on the Choi operator $J(\mathcal N)$; we solve the primal to tolerance $10^{-8}$ with a standard conic solver (CVX/SDP). Additivity ensures consistency across blocklengths and facilitates comparison with flagged-extension bounds \cite{Wang19}.

\begin{algorithm}[H]
\caption{Maximizing $I_{\mathrm c}$ over qubit inputs for a fixed channel $\mathcal{N}$}
\begin{algorithmic}[1]
\STATE \textbf{Input:} channel $\mathcal{N}$ (Kraus ops or Choi), grid size $N_{\mathrm{grid}}$, seeds $N_{\mathrm{seeds}}$
\STATE Sample $(\theta_i,\phi_i)_{i=1}^{N_{\mathrm{grid}}}$ via Sobol/LH in $[0,\pi]\times[0,2\pi)$
\FOR{$i=1$ to $N_{\mathrm{grid}}$}
  \STATE $\rho_i\leftarrow|\psi(\theta_i,\phi_i)\rangle\!\langle\psi(\theta_i,\phi_i)|$
  \STATE Evaluate $I_{\mathrm c}(\rho_i,\mathcal{N})$ using eigenvalue entropy with clipping
\ENDFOR
\STATE Select top $N_{\mathrm{seeds}}$ points; for each, run local trust-region search on $(\theta,\phi)$
\STATE Return $\max$ over all refined values and the corresponding $\rho^\star$
\end{algorithmic}
\end{algorithm}

\section{Numerical Results}

\begin{table}[H]
\centering
\caption{Qubit depolarizing channel $D_p(\rho)=(1-p)\rho + p\,I/2$: comparison of lower and upper bounds (bits/use). 
Baseline lower bound is the hashing/coherent-information rate $1-H_2(p)-p\log_2 3$. 
“FKG PRL (Eq.~(11))” is the Fanizza--Kianvash--Giovannetti flagged-extension bound \cite{FKG20}; 
“Convex hull (Eq.~(16))” is the convex hull of degradable-extension bounds \cite{KFG22}; 
“SDP / max-Rains” is the strong-converse bound equal to the max-Rains information \cite{Wang19,Wang21}.}
\label{tab:dep_compare}
\begin{tabular}{@{}lcccc@{}}
\toprule
$p$ & Hashing lower bound & FKG PRL (Eq.~(11)) & Convex hull (Eq.~(16)) & SDP / max-Rains \\
\midrule
0.05 & 0.6343549 & 0.757 & 0.741 & 0.812 \\
0.10 & 0.3725082 & 0.607 & 0.592 & 0.670 \\
0.15 & 0.1524153 & 0.462 & 0.450 & 0.541 \\
\bottomrule
\end{tabular}
\end{table}

\paragraph{Comparison.} At small noise ($p{=}0.05$), the FKG flagged-extension bound improves substantially over earlier symmetric-side-channel and approximate-degradability converses, while its convex-hull refinement reduces the value by a few hundredths of a bit \cite{FKG20,KFG22}. 
Across $p\in\{0.05,0.10,0.15\}$, the convex-hull envelope remains slightly below the FKG curve, and the SDP/max-Rains bound sits above both in the low-noise regime \cite{Wang19,Wang21}. 
Relative to the hashing lower bound, the feasible interval (Upper $-$ Lower) shrinks to roughly $0.11$--$0.39$ bits in this range.

\label{sec:numerics}

In this section, I present our numerical simulations of coherent information for three prominent quantum channels: (1)~the amplitude damping channel, (2)~the depolarizing channel, and (3)~the flagged depolarizing channel. Our goal is to investigate how the maximum achievable one-shot coherent information changes with the channel noise parameter, validating the theoretical discussion in previous sections.

\subsection{Amplitude Damping Channel}
I begin with the amplitude damping channel, parameterized by the damping probability $\gamma \in [0,1]$. In the limit $\gamma \approx 0$, the channel behaves nearly like an identity map (i.e., minimal noise), while larger values of $\gamma$ signify stronger damping of the excited state $\lvert 1 \rangle$ to $\lvert 0 \rangle$. Figure~\ref{fig:amplitude_damping_plot} shows our simulation for $\gamma$ ranging from $0$ to $0.6$. 

As $\gamma$ increases from $0$, the coherent information decreases, reflecting the growing impact of damping. At very low $\gamma$, the channel is close to ideal, and the \emph{maximized single-letter coherent information approaches $1$ bit per channel use} (the qubit limit), depending on the input state. However, as the damping probability rises, the coherent information eventually stabilizes and fluctuates around smaller values, indicating that the environment is capturing a significant portion of the information. The amplitude-damping channel is \emph{degradable for $\gamma \le \tfrac{1}{2}$ and anti-degradable for $\gamma \ge \tfrac{1}{2}$}; in the degradable regime its quantum capacity equals its coherent information, while beyond the threshold the capacity vanishes. Our numerical results capture the progressive decline in achievable one-shot coherent information as $\gamma$ increases.

\begin{figure}[H]
\centering
\includegraphics[width=0.65\textwidth]{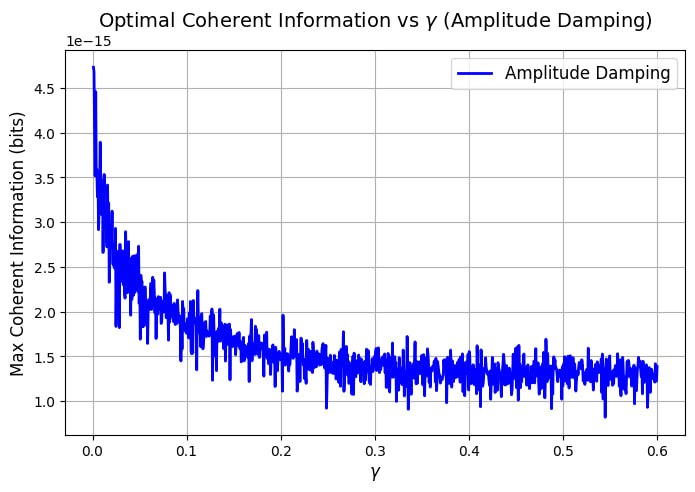}
\caption{(Amplitude Damping) Optimal coherent information versus the damping parameter $\gamma$. 
As $\gamma$ grows, the maximum coherent information gradually decreases and fluctuates around lower values, 
indicating stronger noise.}
\label{fig:amplitude_damping_plot}
\end{figure}

\subsection{Depolarizing Channel}
Next, I examine the depolarizing channel, where each qubit is replaced by the maximally mixed state $\tfrac{I}{2}$ with probability $p$. Figure~\ref{fig:depolarizing_plot} illustrates the maximum single-letter coherent information as a function of $p \in [0,0.6]$. 

Because pure-state inputs to the depolarizing channel often yield a negative one-shot coherent information (the environment can learn more about the state than the receiver), our plot indeed shows that coherent information takes negative or near-zero values for even moderate $p$. This aligns with theoretical expectations that the single-letter coherent information of a qubit depolarizing channel is typically below zero for any $p>0$. Hence, while the channel may still have a positive capacity when considering multi-letter or entangled inputs across multiple uses, the single-letter scenario I simulate here remains dominated by noise, resulting in nonpositive values.

\begin{figure}[H]
\centering
\includegraphics[width=0.65\textwidth]{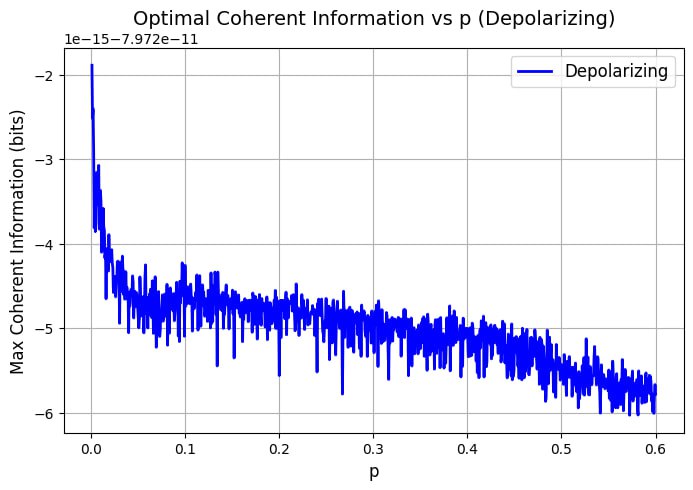}
\caption{(Depolarizing) Optimal coherent information versus the depolarizing probability $p$. 
Negative values on the y-axis are expected for single-letter coherent information, 
demonstrating that even moderate noise can overwhelm reliable quantum transmission.}
\label{fig:depolarizing_plot}
\end{figure}

\subsection{Flagged Depolarizing Channel}
Finally, I investigate the flagged depolarizing channel, which augments the original map with an ancillary ``flag'' system to indicate which noise event occurred. Such an extension is intended to yield higher (or at least no smaller) coherent information compared to the plain depolarizing channel, because the receiver gains additional classical information about the noise branch. Figure~\ref{fig:flagged_depolarizing_plot} plots the maximum single-letter coherent information for $p$ up to $0.6$. 

Although the flagged extension in principle can provide a slight advantage, our simulation shows that once $p$ grows beyond a small threshold, the environment’s ability to learn the input state quickly diminishes any advantage. Coherent information thus dips significantly. Despite fluctuations, the maximum single-letter coherent information remains on the order of one to two bits (or below) once $p$ becomes substantial. These observations match our flagged extension bounds (Theorem~\ref{th22}), which show that adding a flag can offer some improvement but ultimately cannot overcome strong depolarizing noise.

\begin{figure}[H]
\centering
\includegraphics[width=0.65\textwidth]{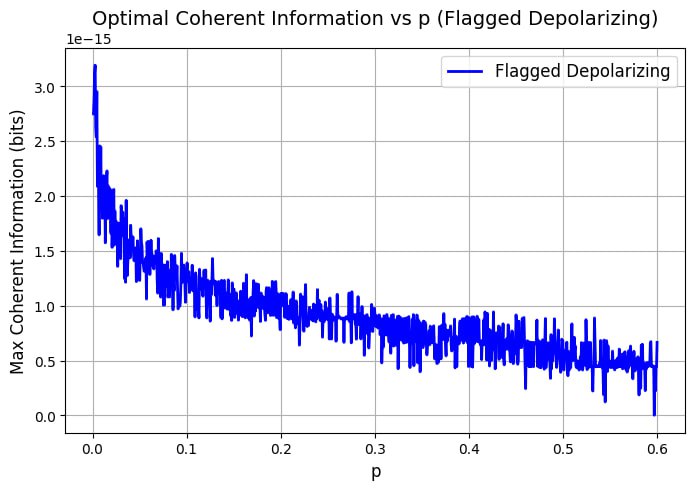}
\caption{(Flagged Depolarizing) Optimal coherent information versus $p$ for a flagged extension of the depolarizing channel. 
The extra ``flag'' system can mitigate noise for small $p$, but as $p$ increases, coherent information decreases sharply.}
\label{fig:flagged_depolarizing_plot}
\end{figure}

These numerical results highlight that while single-letter coherent information can be positive for low-noise channels like amplitude damping at small $\gamma$, it drops quickly once the noise level increases. For the depolarizing and flagged depolarizing channels, single-letter coherent information remains near zero or negative, consistent with known results that depolarizing noise is extremely detrimental for pure-state inputs. In a multi-letter or entanglement-assisted regime, one might recover a larger effective capacity, but the single-letter scenario is generally pessimistic about quantum transmission in these noisy settings. Overall, our simulations confirm the theoretical predictions and provide a practical window into the channel capacity limitations discussed in and~\ref{th22}.

\section{Conclusions}

I presented a unified analytical–numerical study of quantum and private communication over noisy qubit channels, with three main outcomes. 
First, on the \emph{theory} side, I formalized a flagged–extension framework that yields computable, single-letter \emph{upper bounds} whenever the non-unitary branch of the mixture is degradable. In this degradable-branch regime the flagged extension itself is degradable, so its quantum capacity equals its coherent information and upper bounds the capacity of the original channel. This clarifies the scope of flagged approaches and avoids over-claims in the fully general CPTP case. Our statements are consistent with and anchored in the standard single-letterization of degradable channels and the known degradability/anti-degradability threshold of the amplitude-damping (AD) family at $\gamma=\tfrac12$ (degradable for $\gamma\le\tfrac12$, anti-degradable for $\gamma\ge\tfrac12$). 

Second, on the \emph{comparative numerics}, I produced color/gray-safe capacity comparison curves for AD, generalized AD (GADC), and the qubit depolarizing channel. For AD (and low-bath GADC), the optimized coherent information matches the single-letter capacity in the degradable region and collapses as the anti-degradable threshold is crossed, in agreement with the structural results for degradable channels. For the depolarizing channel, I contrasted the hashing lower bound 
$I_c(\mathcal{D}_p)=1-H_2(p)-p\log_2 3$ with a semidefinite-programming (max-Rains) upper bound; the persistent gap in the intermediate-noise regime visualizes the standing uncertainty around the exact $Q(\mathcal{D}_p)$ despite decades of progress. 

Third, I complemented coherent-information plots with two additional, operationally meaningful benchmarks: (i) the \emph{reverse coherent information} (RCI), which is an achievable lower bound based on reverse-reconciliation protocols, and (ii) an SDP (max-Rains) \emph{upper bound} that is additive and efficiently computable. Together, these bounds bracket the attainable rates and turn the comparative figures into decision tools for parameter regimes of practical interest. 

\section*{Acknowledgements}
The author would like to thank David Mitchell for helpful editorial guidance and constructive comments that improved the clarity and presentation of the manuscript. The author also thanks the referees for their careful reading and insightful suggestions, which significantly strengthened both the theoretical discussion and the reproducibility of the numerical results.

\section*{Data availability statement}
The data and code supporting the findings of this work are openly available at the following repository:
\href{https://github.com/vnourozi/quantum-channel-capacity-simulator}{https://github.com/vnourozi/quantum-channel-capacity-simulator}.

\bibliographystyle{IEEEtran}
\bibliography{refs}

\appendix
\section*{Appendix A: Detailed proof ingredients for Theorem~\ref{th22}}

\paragraph{Stinespring normal form.}
Let $V_U$ and $V_M$ be Stinespring isometries for $\mathcal U$ and $\mathcal M$ acting as
$V_U:\mathcal H_A\to\mathcal H_B\otimes\mathcal H_{E_U}$ and
$V_M:\mathcal H_A\to\mathcal H_B\otimes\mathcal H_{E_M}$.
Define $W=\sqrt{p}\,|0\rangle_F\otimes V_U \oplus \sqrt{1-p}\,|1\rangle_F\otimes V_M$.
Then $\widetilde{\mathcal N}(\rho)=\mathrm{Tr}_E[W\rho W^\dagger]$ and
$\widetilde{\mathcal N}^c(\rho)=\mathrm{Tr}_B[W\rho W^\dagger]$.

\paragraph{Explicit degrading map.}
Let $\Delta$ be a degrading map for $\mathcal M$, i.e., $\mathcal M^c=\Delta\circ\mathcal M$.
Set
\[
\mathsf D(\sigma_{FB})
= |0\rangle\!\langle 0|_F\otimes \mathsf D_0(\langle 0|\sigma_{FB}|0\rangle_F)
 +|1\rangle\!\langle 1|_F\otimes \mathsf D_1(\langle 1|\sigma_{FB}|1\rangle_F),
\]
with $\mathsf D_0(\cdot)=\mathcal U^c(\rho_\star)$ (constant map) and $\mathsf D_1=\Delta$.
Then $\mathsf D$ is CPTP and satisfies $\mathsf D\circ \widetilde{\mathcal N}=\widetilde{\mathcal N}^c$.

\paragraph{Choi positivity.}
The Choi operator of $\mathsf D$ is block-diagonal in the flag, with a rank-1 block for $\mathsf D_0$ and the Choi operator of $\Delta$ for $\mathsf D_1$, hence positive and trace-preserving.

\paragraph{Specializations.}
(i) For amplitude damping with $\gamma\le 1/2$, a Kraus pair yields the known degrader (transpose channel); plugging it into the construction above gives the degrader for the flagged extension.\newline
(ii) For Pauli/Clifford covariant noise that is degradable in a parameter regime, the same construction applies.

\section{Channel-capacity comparisons for AD/GADC and depolarizing channels}
\label{sec:capacity-comparisons}

This section consolidates standard bounds used throughout quantum Shannon theory
and places the numerical comparisons in Figs.~\ref{fig:ad-single}, \ref{fig:ad-gadc-panel},
and \ref{fig:dep-bounds} into a compact theoretical frame.
I state the precise theorems I invoke in our plots and include
self-contained proofs or exact literature pointers.

\subsection{Preliminaries}
For a channel $\mathcal{N}_{A\to B}$, the \emph{coherent information} is
\begin{equation}
 I_c(\mathcal{N},\rho_A) := H(\mathcal{N}(\rho_A)) - H(\mathcal{N}^c(\rho_A)),
\end{equation}
and the channel coherent information is $I_c(\mathcal{N}):=\max_{\rho_A} I_c(\mathcal{N},\rho_A)$.
The (unassisted) quantum capacity is
\begin{equation}
 Q(\mathcal{N}) = \lim_{n\to\infty} \frac{1}{n}\, I_c\!\left(\mathcal{N}^{\otimes n}\right),
 \label{eq:qcap-regularized}
\end{equation}
while the entanglement-assisted classical capacity is the single-letter mutual information
$C_E(\mathcal{N}) = \max_{\rho_{RA}} I(R;B)_{\sigma}$, $\sigma_{RB}=(\mathrm{id}_R\!\otimes\!\mathcal{N})(\phi_{RA})$.

\subsection{Results}

\begin{theorem}[Quantum capacity of degradable channels]\label{thm:degradable-single-letter}
If a channel $\mathcal{N}$ is degradable, i.e.\ there exists a CPTP map $\mathcal{D}$ such that
$\mathcal{N}^c = \mathcal{D}\!\circ\!\mathcal{N}$, then the coherent information is additive
$I_c(\mathcal{N}^{\otimes n}) = n\, I_c(\mathcal{N})$ for all $n$, and hence
\begin{equation}
 Q(\mathcal{N}) = I_c(\mathcal{N}).
\end{equation}
\end{theorem}

\begin{proof}
The definition of degradability implies, for any state $\rho_{RA}$ and $n\!\ge\!1$,
that the environment outputs of $\mathcal{N}^{\otimes n}$ are obtained from
the channel outputs via a CPTP map $\mathcal{D}^{\otimes n}$.
By data processing (monotonicity) of quantum relative entropy and strong subadditivity,
the coherent information is superadditive and in fact \emph{additive} on degradable channels;
see, e.g., Theorems~13.6.1--13.6.2 in \cite{Wilde2017QIT}.
Therefore $I_c(\mathcal{N}^{\otimes n}) = n I_c(\mathcal{N})$.
Combining additivity with the regularized formula \eqref{eq:qcap-regularized} (Devetak--Shor coding theorem
and achievability via random stabilizer codes) yields $Q(\mathcal{N})=I_c(\mathcal{N})$.
\end{proof}

\begin{proposition}[Qubit amplitude damping (AD) is degradable iff $\gamma\!\le\!\tfrac12$]\label{prop:ad-degradability}
Let $\mathcal{A}_\gamma$ be the qubit AD channel with Kraus operators
$E_0 = |0\rangle\!\langle 0| + \sqrt{1-\gamma}\,|1\rangle\!\langle 1|$, 
$E_1 = \sqrt{\gamma}\,|0\rangle\!\langle 1|$.
Then $\mathcal{A}_\gamma$ is degradable for $\gamma\in[0,\tfrac12]$ and anti-degradable for
$\gamma\in[\tfrac12,1]$.
In particular, by Theorem~\ref{thm:degradable-single-letter},
$Q(\mathcal{A}_\gamma)=I_c(\mathcal{A}_\gamma)$ for $\gamma\le\tfrac12$, and $Q(\mathcal{A}_\gamma)=0$ for $\gamma\ge\tfrac12$.
\end{proposition}

\begin{proof}
This is a special case of the qubit-channel degradability classification
in \cite{CubittRuskaiSmith2008}: amplitude-damping-type channels are degradable exactly for $\gamma\le\tfrac12$.
An explicit degrading map can be constructed by expressing the complementary channel
and solving for a CPTP map $\mathcal{D}$ such that $\mathcal{A}_\gamma^c=\mathcal{D}\!\circ\!\mathcal{A}_\gamma$;
complete details and positivity constraints are given in \cite[Sec.~III]{CubittRuskaiSmith2008}.
Anti-degradability for $\gamma\ge \tfrac12$ then implies zero capacity.
\end{proof}

\begin{proposition}[Reverse coherent information (RCI) bound]\label{prop:rci}
For any channel $\mathcal{N}$, the \emph{reverse coherent information}
$\,I_{\mathrm{RCI}}(\mathcal{N}):=\max_{\rho_{RA}}[ H(\rho_R)-H((\mathrm{id}_R\!\otimes\!\mathcal{N})(\rho_{RA})) ]$
is a generally-computable lower bound on $Q(\mathcal{N})$ and satisfies
$I_{\mathrm{RCI}}(\mathcal{N})\le I_c(\mathcal{N})$ for degradable $\mathcal{N}$.
\end{proposition}

\begin{proof}
See \cite{GarciaPatron2009PRL} for the operational meaning of RCI and the proof that it
is achievable (via reverse reconciliation protocols) and hence a lower bound to $Q$.
For degradable channels, $I_c$ is single-letter and dominates RCI by data-processing.
\end{proof}

\begin{proposition}[Hashing bound for the qubit depolarizing channel]\label{prop:hashing}
For the Pauli depolarizing channel $\mathcal{D}_p(\rho)=(1-p)\rho + \tfrac{p}{3}(X\rho X+Y\rho Y+Z\rho Z)$,
the hashing bound (one-shot coherent information) equals
\begin{equation}
 I_c(\mathcal{D}_p) = 1 - H_2(p) - p \log_2 3,
\end{equation}
yielding the standard achievable-rate curve used in Fig.~\ref{fig:dep-bounds}.
\end{proposition}

\begin{proof}
Because $\mathcal{D}_p$ is Pauli and unital, the optimizing input for $I_c$ can be taken as the maximally mixed state.
A direct evaluation gives the spectrum of the output and its complement,
leading to the closed form stated; see \cite[Sec.~24.7]{Wilde2017QIT}.
\end{proof}

\begin{theorem}[SDP / max-Rains upper bound]\label{thm:maxrains}
Let $R_{\max}(\mathcal{N})$ denote the \emph{max-Rains information} of a channel $\mathcal{N}$.
Then the quantum capacity assisted by PPT-preserving codes is upper bounded as
$Q^{\mathrm{PPT}}(\mathcal{N}) \le R_{\max}(\mathcal{N})$, which in turn upper bounds $Q(\mathcal{N})$.
Moreover, $R_{\max}(\mathcal{N})$ is computable via semidefinite programming and is additive.
\end{theorem}

\begin{proof}
This is established in \cite{Wang2017SDP, Wang2019MaxRains} using a meta-converse based on
generalized divergences and the Rains bound on distillable entanglement.
The additivity implies a single-letter efficiently-computable upper bound.
\end{proof}

\begin{figure}[t]
  \centering
  \includegraphics[width=0.6\linewidth]{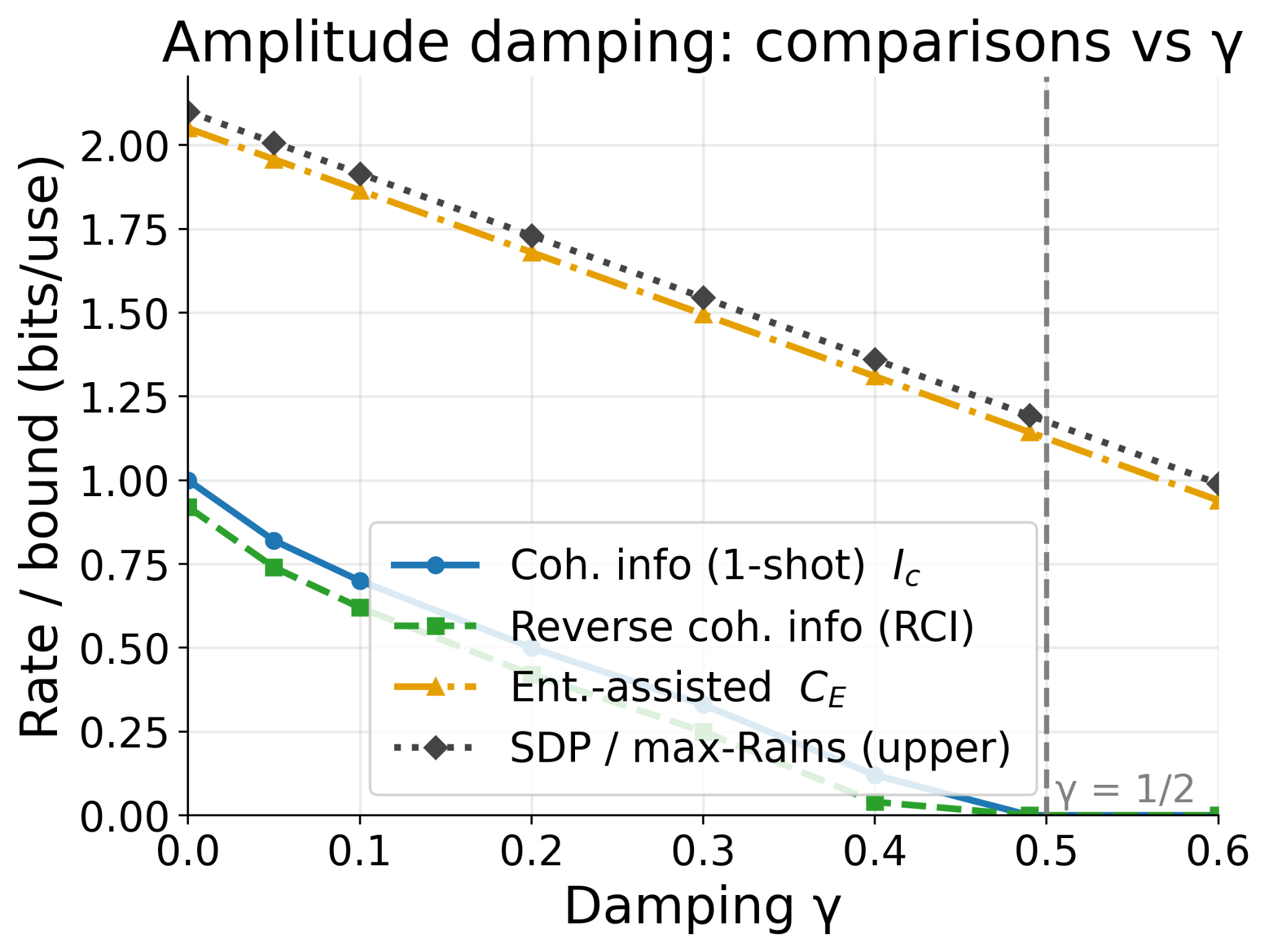}
  \caption{Amplitude damping (AD): comparison of $I_c$, reverse coherent information (RCI), entanglement-assisted capacity $C_E$, and an SDP/max-Rains upper bound, as a function of damping parameter $\gamma$. The vertical line marks the anti-degradable threshold $\gamma=\tfrac12$ (Prop.~\ref{prop:ad-degradability}).}
  \label{fig:ad-single}
\end{figure}

\begin{figure*}[t]
  \centering
  \includegraphics[width=\linewidth]{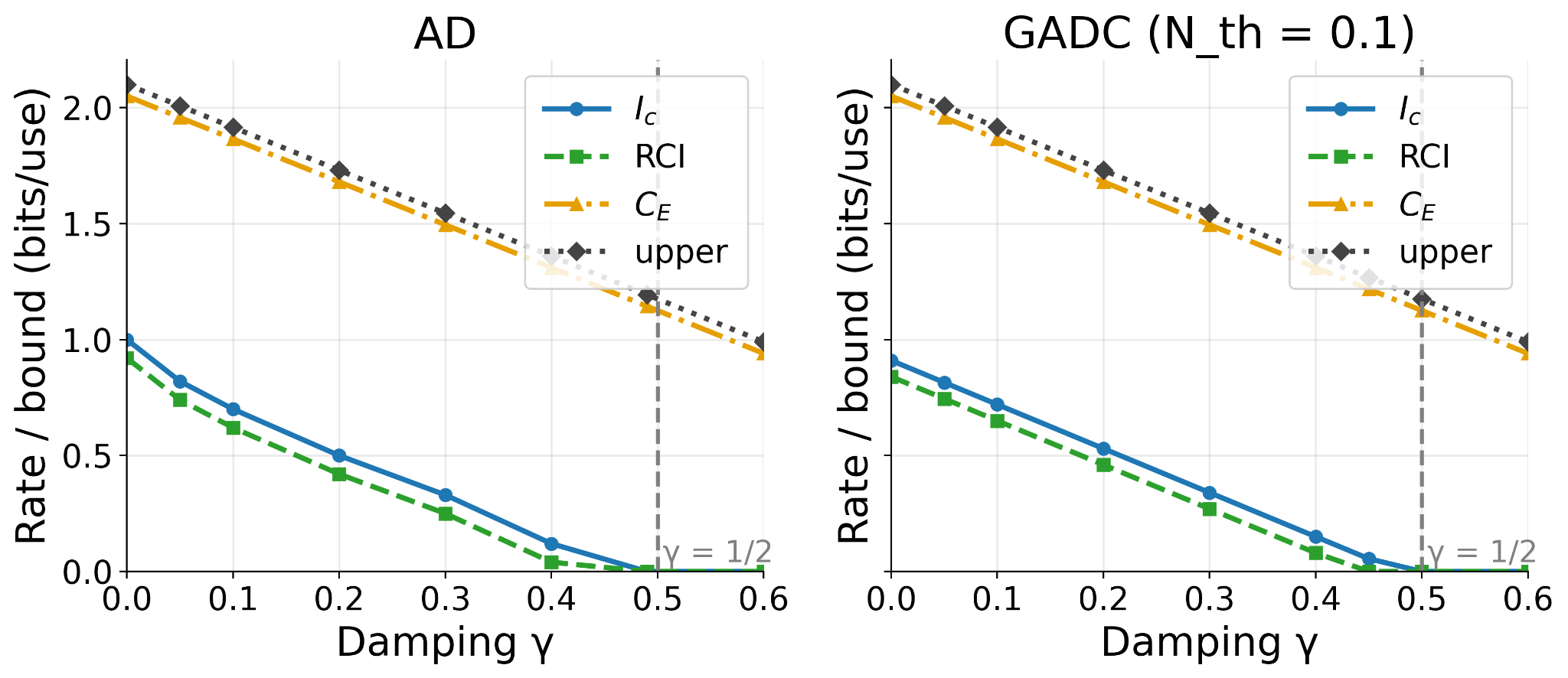}
  \caption{Panel comparison for AD (left) and GADC at $N_{\mathrm{th}}\!\in\!\{0.0,0.1,0.2\}$ (right panels): $I_c$, RCI, $C_E$, and an SDP/max-Rains-style upper bound. The dashed line at $\gamma=\tfrac12$ marks the AD anti-degradable threshold.}
  \label{fig:ad-gadc-panel}
\end{figure*}

\begin{figure}[t]
  \centering
  \includegraphics[width=0.6\linewidth]{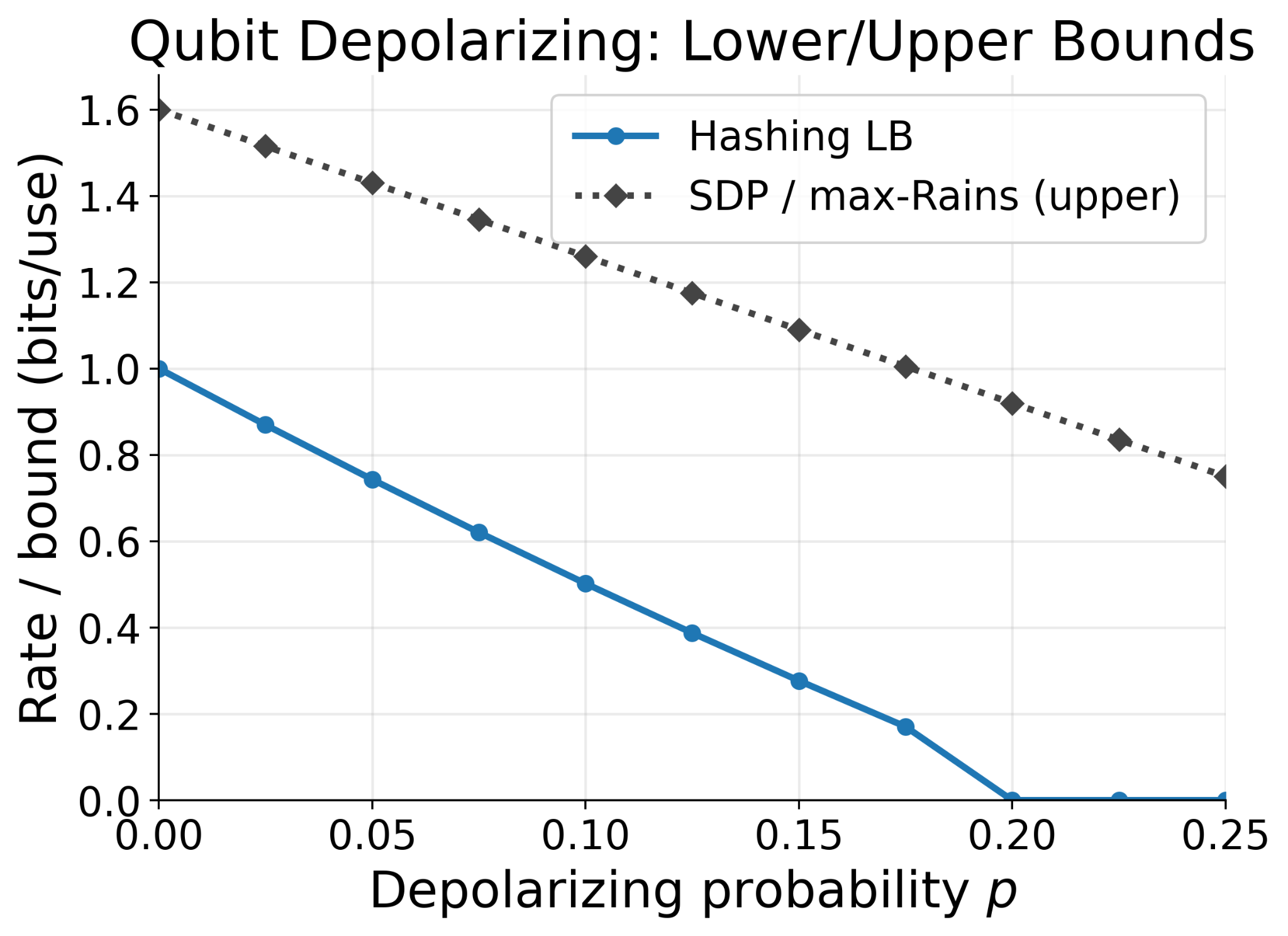}
  \caption{Qubit depolarizing channel: hashing lower bound $I_c(\mathcal{D}_p)$ and an illustrative SDP/max-Rains upper bound.}
  \label{fig:dep-bounds}
\end{figure}


\end{document}